%% file: ccn-pitless.tex
\documentclass[10pt,conference,compsocconf,letterpaper]{IEEEtran}

\usepackage{color, colortbl}
\usepackage{multirow}
\usepackage{transparent}
\usepackage{balance}
\usepackage{url}
\usepackage{graphicx}
\usepackage{paralist}
\usepackage{amssymb, framed, graphics}
\usepackage{bytefield}
\usepackage{todonotes}
\usepackage{algorithm}
\usepackage{algpseudocode}
\usepackage[group-separator={,}]{siunitx}
\usepackage[underline=true]{pgf-umlsd}
\usepackage{hyphenat}
\usepackage{subfigure}
\usepackage{amsmath}
\usepackage{cite}
\usepackage{fancybox,fancyvrb}
\usepackage{url}
\usepackage{framed}
\usepackage{xspace}
\usepackage{subfigure}

\newtheorem{theorem}{Theorem}

\newcommand{\ccnname}[1]{{\path{#1}}}

\newcommand{\pitless}{\textsc{PIT-less}}

\begin{document}
\title{Living in a \pitless\ World: A Case Against Stateful Forwarding in Content-Centric Networking}
\vspace{-1em}
\author{\IEEEauthorblockN{Cesar Ghali \quad\quad\quad Gene Tsudik}
\IEEEauthorblockA{University of California, Irvine\\
\{cghali, gene.tsudik\}@uci.edu}
\and
\IEEEauthorblockN{Ersin Uzun}
\IEEEauthorblockA{Palo Alto Research Center\\
ersin.uzun@parc.com}
\and
\IEEEauthorblockN{Christopher A. Wood}
\IEEEauthorblockA{University of California, Irvine\\
woodc1@uci.edu}}

\maketitle

\begin{abstract}
Information-Centric Networking (ICN) is a recent paradigm that claims
to mitigate some limitations of the current IP-based Internet
architecture. The centerpiece of ICN is named and addressable content,
rather than hosts or interfaces. Content-Centric Networking (CCN)
is a prominent ICN instance that shares the fundamental architectural
design with its equally popular academic sibling Named-Data
Networking (NDN). CCN eschews source addresses and creates
one-time virtual circuits for every content request (called an
interest). As an interest is forwarded it creates state in
intervening routers and the requested content back is delivered
over the reverse path using that state.

Although a stateful forwarding plane might be beneficial in terms
of efficiency, and resilience to certain types of attacks,
this has not been decisively proven via realistic experiments.
Since keeping per-interest state complicates router operations
and makes the infrastructure susceptible to router state
exhaustion attacks (e.g., there is currently no effective defense
against interest flooding attacks), the value of the
stateful forwarding plane in CCN should be re-examined.

In this paper, we explore supposed benefits and various problems
of the stateful forwarding plane. We then argue that its benefits are
uncertain at best and it should not be a mandatory CCN feature.
To this end, we propose a new stateless architecture for CCN
that provides nearly all functionality of the stateful design
without its headaches. We analyze performance and resource
requirements of the proposed architecture, via experiments.
\end{abstract}

\input{01-introduction}
\input{02-pit}
\input{03-pitless}
\input{04-assessment}
\input{05-experiments}
\input{06-related}

\section{Conclusion and Future Work}\label{sec:conclusion}
Motivated by the Interest Flooding attacks in current CCN, 
we proposed an alternative CCN architecture without PITs, called
stateless CCN. We investigated the benefits of PIT and realized that they do not
significantly improve the performance of content distribution. Our proposed
architecture is based on Routable Backward Names (RBNs) used to route content
back towards requesting consumers. We provided a comprehensive
performance and security assessment of the proposed stateless CCN architecture.
We also discussed how it is practical to deploy this architecture in today's IP
networks and showed that deploying it alongside with current CCN does not achieve
the expected benefits and performance.

However, removing the PIT came at the expense of losing support of CCN features
and extensions developed throughout the last few years. Consumer anonymity, for
instance, cannot be achieved in RBN-based stateless CCN at the network layer
without using supporting protocols such as {\sf AND\=aNA} \cite{dibenedetto2012andana}.
Moreover, the Interest-Key
Binding rule (IKB) \cite{ghali2014network} that allows efficient
content trust at the network layer relies heavily on PIT. Thus, the IKB rule cannot
be applied in RBN-based stateless CCN. We believe that the advantages of the
proposed CCN architecture outweigh its drawbacks. We therefore defer solutions to
the aforementioned disadvantages, e.g., anonymity and trust, to future work.

\balance

\bibliographystyle{IEEEtran}
\bibliography{IEEEabrv,references}

\end{document}

%% file: 01-introduction.tex
\section{Introduction}
%

Information-Centric Networking (ICN) \cite{ahlgren2012survey} is a networking
model that emerged as an alternative to the host-based communication approach of 
the current IP-based Internet
architecture. Content-Centric Networking (CCN) \cite{jacobson2009networking,ccn10}
is one industry-driven instance of this model. (It is closely related to Named-Data Networking (NDN),
which can be viewed as CCN's academic dual.) While IP traffic consists of packets sent
between communicating end-points, CCN traffic is comprised of explicit requests
for, and responses to, named content objects. These requests, called {\em interest} messages,
refer to the desired content by name. An interest is forwarded by routers (using the name) 
towards a content producer until satisfied by the latter or by a cached copy in some router.
The corresponding response, called {\em content}, is forwarded along
the reverse path. To reduce end-to-end latency and congestion, CCN routers may
opportunistically cache content to satisfy future interests.


In CCN, neither interest nor content messages carry source addresses. In order to correctly deliver
content to consumers, routers maintain per-interest state as an entry in a so-called Pending Interest
Table (PIT). This state information maps interest names to interfaces on
which they arrived. This enables a router which receives a content easily identify
the forwarding interface(s) using the corresponding PIT entry. Once a content
is forwarded downstream, the corresponding PIT entry is flushed. 

Another purpose of the PIT is to support {\em interest collapsing} -- a feature
useful for handling multitudes of nearly simultaneous interests for the same content.
Whenever a router receives an interest for which it has a matching PIT entry, the arrival interface of
the new interest is added to the existing entry and the interest is not forwarded
further. This prevents duplicate interests from being sent upstream, thus lowering overall
congestion. However, as we show later in the paper, interest collapsing rarely occurs in practice.

Furthermore, stateful forwarding enabled by PITs is supposed to provide flow balance
via path symmetry between interest and content messages. Consequently,
information from the PIT (e.g., interest to content Round-Trip Time (RTT))
can be used to develop better 
congestion control and traffic shaping mechanisms; see \cite{yi2012adaptive,zhou2015proactive,park2014popularity,wang2013improved}.
However, using a PIT for flow balance and in-network congestion
control is quite problematic in practice. In fact, flow balance is a false claim
in the current CCN design due to the (potentially huge) disparity in sizes between interest and content
messages. Likewise, there is ample evidence that congestion control and transport protocols 
are best deployed at the receiver \cite{carofiglio2013multipath,braun2013empirical,saino2013cctcp},
due to flow imbalance and dynamic routing in CCN. Another claimed advantage of
stateful forwarding is that it can aid routers when responding to network problems since
they can make autonomous and intelligent forwarding decisions for interests.
In practice, however, individual routers rarely have sufficient autonomy to make such decisions.

From the perspective of infrastructure security, the PIT effectively prevents
{\em reflection attacks} since content is always forwarded according
to PIT entries \cite{gasti2013and}. However, such attacks can be mitigated by
mechanisms that do not require any forwarding state. Moreover, this state
is costly to maintain. Various attempts to improve the efficiency of PIT-based
forwarding have been studied in the context of CCN and 
NDN~\cite{so2012toward,tsilopoulos2014reducing,yuan2014scalable}.
However, they do not address the fundamental design issue that the PIT size
grows linearly with the number of distinct interests received by a router. This means
that a PIT is a resource that can be easily abused. In fact, malicious exhaustion of PIT space
in the form of Interest Flooding (IF) attacks \cite{gasti2013and} remains an important 
open problem. In such attacks, adversaries can flood routers with nonsensical (i.e., unsatisfiable)
or slow-to-satisfy (i.e., requiring dynamic content generation) interests in order
to maximize the occupied PIT space. Once a router reaches its maximum PIT capacity it either:
(1) drops new incoming interests, or (2) removes existing entries to free resources
for new incoming interests. Unfortunately, these two options result in effective DoS 
for legitimate future or current interests, respectively.

Given that many of the claimed benefits are dubious and considering 
associated the infrastructure security problems, it becoms hard to justify the need
for PITs in CCN. Therefore, in this paper, we comprehensively
assess (in Section \ref{sec:pit}) the stateful forwarding plane of CCN with respect to
each claimed benefit. We show that such benefits are: (1) either 
unrealistic or infeasible in practice, (2) can be achieved 
by means other than stateful forwarding, or (3) so marginal
that their worth simply does not justify the overhead. We then present,
in Sections \ref{sec:pitless} and \ref{sec:assessment}, a new stateless architecture for
CCN based on Routable Backward Names (RBNs). This new design can co-exist with the
current CCN architecture (with PITs) or replace it entirely. Experimental
results in Section \ref{sec:experiment} indicate that the new design
still retains the essence and performance characteristics
of CCN while successfully avoiding pitfalls of stateful forwarding.
We conclude with a discussion of related work and a summary in
Sections \ref{sec:related} and \ref{sec:conclusion}, respectively.

%% file: 02-pit.tex
\section{Assessing the PIT}\label{sec:pit}
The PIT is a fundamental and mandatory feature of the CCN forwarding plane.
It is a tabular data structure that maps interest names and other
metadata to a set of information, such as arrival interfaces and lifetime values.
Arrival interfaces are used to identify downstream interfaces on which content
responses should be forwarded.
The shape and size of this table is directly dependent on the traffic that is processed
by a forwarder. \cite{carofiglio2015pending} studied dynamics
of the PIT and showed that the number of entries can range from less than $100$
for edge routers with a small number of per-namespace flows to over $10^6$
in the core. \cite{yuan2014scalable} designed a PIT implementation that
requires only 37MiB to 245MiB to forward traffic at 100Gbps, which can scale to
fit the needs of realistic traffic, according to \cite{carofiglio2015pending}.

However, in this paper, we do not question the implementaton of the PIT.
Instead, we question its entire existence. Below, we argue that aside from being unecessary to support 
CCN-like communication, the PIT's presence raises more (serious) problems than it solves. 
We support our argument by systematically analyzing the following alleged PIT benefits:
\begin{compactenum}
    \item Content object forwarding
    \item Interest collapsing
    \item Flow and congestion control
    \item Infrastructure security
\end{compactenum}
We then show that all these benefits are either false, unnecessary, or very meager at best.

\subsection{Content Object Forwarding}\label{sec:content-forwarding}
A key tenet of CCN is that content is never sent to a consumer who did not previously 
issue an interest for the (name of that) content. 
Since interests contain no source addresses, PITs
are needed ``[t]o forward Content Objects from producers to consumers along the
interest reverse path by leaving per-hop state in each router...'' \cite{jacobson2009networking}.

We disagree with this statement for two reasons. First, network path symmetry
is not guaranteed and should not be assumed. Wolfgang et al. \cite{john2010estimating} 
demonstrated that route symmetry between the same flow on the Internet is lower in the core
than at the edges. Several tier-1 and tier-2 networks were studied and it was shown 
that, due to ``hot-potato-routing,'' flow asymmetry exceeds $90\%$ in the core. 
Thus, symmetric path routing in the core appears to direct contradict today's practices 
that promote path asymmetry. Attempting to enforce symmetric data traversal appears
to be a challenge from an economic perspective.

Second, pull-based communication with symmetric paths is not well-suited for
\emph{all} applications. While appropriate for scalable content distribution 
applications\footnote{Which some believe to be already well-served by today's CDNs.},
it is substantially different from modern TCP/IP applications and protocols which rely on 
point-to-point bi-directional streams between endpoints. For instance, the WebSocket 
\cite{fette2011websocket} protocol uses full-duplex TCP streams for clients and servers that
engage in real-time, bidirectional communication. It is used by many popular 
interactive applications, such as multimedia chat and multiplayer video games. 
Two-way communication is not limited to Web protocols. Voice applications such
as Skype \cite{skype} and peer-to-peer systems such as BitTorrent \cite{cohen2008bittorrent}
rely on two endpoints which both produce and consume data, as part of the application.

Given the relative infancy of CCN and abundance of real-world applications that currently 
do not fit CCN's mold,  it is difficult to argue that the pull communication model can
satisfy all applications' needs. For example, even some existing CCN applications exploit interest
messages to carry information from consumers to producers \cite{burke2013securing}.
Other applications rely on consumers and producers to send interests to
each other. NDN-RTC, a recently developed NDN video teleconference application, is one
specific example that supports such bidirectional communication between peers \cite{gusev2015ndn}.
(We use NDN and CCN interchangeably here since both are equivalent in this
context.)
Another emerging application design pattern is  data transport via
set synchronization, notably, the NDN ChronoSync protocol \cite{zhu2013let}. 
It enables data synchronization among a set of users. 
Each ChronoSync user acts as \emph{both} a producer and consumer. Consumers
(members) issue long-standing interests to a group (common namespace) about specific data
to be synchronized, that are routed to all members. When target
data is changed by someone, this member satisfies previous interest(s) with a 
fingerprint of the data in a content object. Each member is then responsible for 
requesting updated content to synchronize with the others.

Based on the trends of current TCP/IP applications and proposed design strategies for
CCN-based protocols and applications, it seems clear that bidirectional communication
is here to stay. For it to work, router FIBs need to contain prefixes for all end-points -- 
not just producers. Therefore, all communicating parties need to obtain and 
use a routable prefix, which effectively serves as an address. As a consequence,
forwarding information stored in a PIT becomes redundant and unnecessary.

\subsection{Uitility of Interest Collapsing}
Recall that collapsing applies only to interests arriving at routers during
a very small time-window $\Delta$: between the time of arrival of the original interest
(referring to the same name) that triggered creation of a new PIT entry, and the time of 
arrival of the corresponding content. Due to increasing network data rates and
lower end-to-end delays facilitated by in-network caching, the value of
$\Delta$ is expected to be miniscule, e.g., on order of tens of milliseconds. 
Therefore, we believe that the effect of interest collapsing would not play 
any significant role in the real life performance of CCN.

To support this claim, we model the probability of interest collapsing occurring
in the first-hop router $R$. The reason for choosing the first router is that the
benefits of interest collapsing are felt the most closest to the consumer(s). This is because
collapsing two similar interests at the consumer-facing router reduces bandwidth
usage more than when the same occurs closer to the producer. We assume that
content popularity follows a Zipf distribution with classes $k=1,\dots,K$ and average
number of segments $\sigma_k$.\footnote{Large content is typically split into smaller segments.}
Let each class arrival rate $\lambda_k$ at $R$ be modeled as a Poisson
process. The event of interest collapsing at $R$ for content class $k$ is denoted
as $\textsf{Coll}_{int}^R(k)$. The probability of this event is \cite{carofiglio2011modeling}:
\begin{align}
\label{equ:int_coll_1}
\Pr \left[ \textsf{Coll}_{int}^R(k) \right] =
			\frac{1 - e^{-\Delta \lambda_k}}{1 - (1 - 1/\sigma_k)e^{-\Delta \lambda_k}}
\end{align}

\begin{figure*}[t]
\centering
\subfigure[Cache hit rates for all classes of content is 0 ($R$'s cache is disabled).]
{
	\includegraphics[width=\columnwidth]{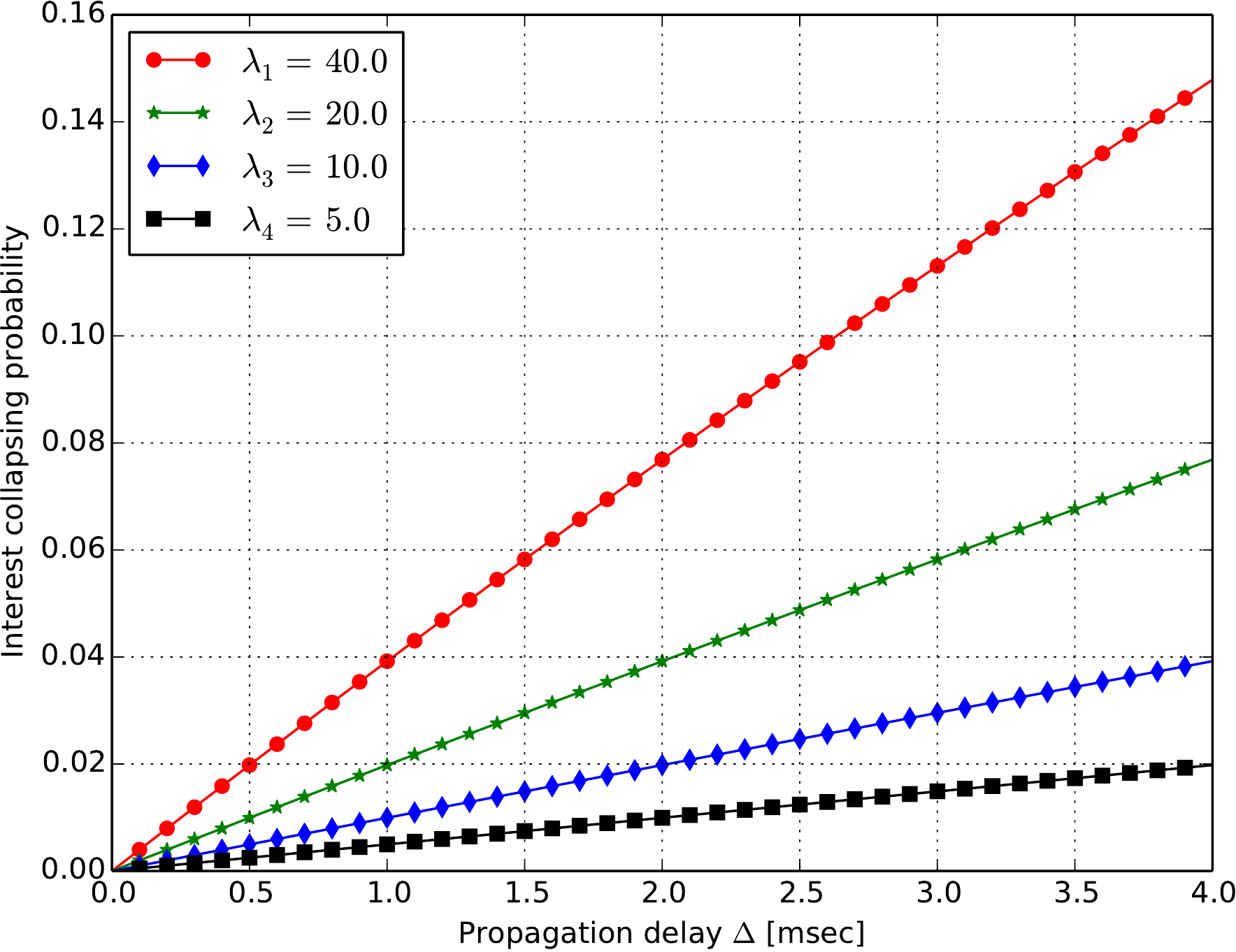}
	\label{fig:int_coll_no_cache}
}
\subfigure[Cache hit rates differ for different classes of content.]
{
	\includegraphics[width=\columnwidth]{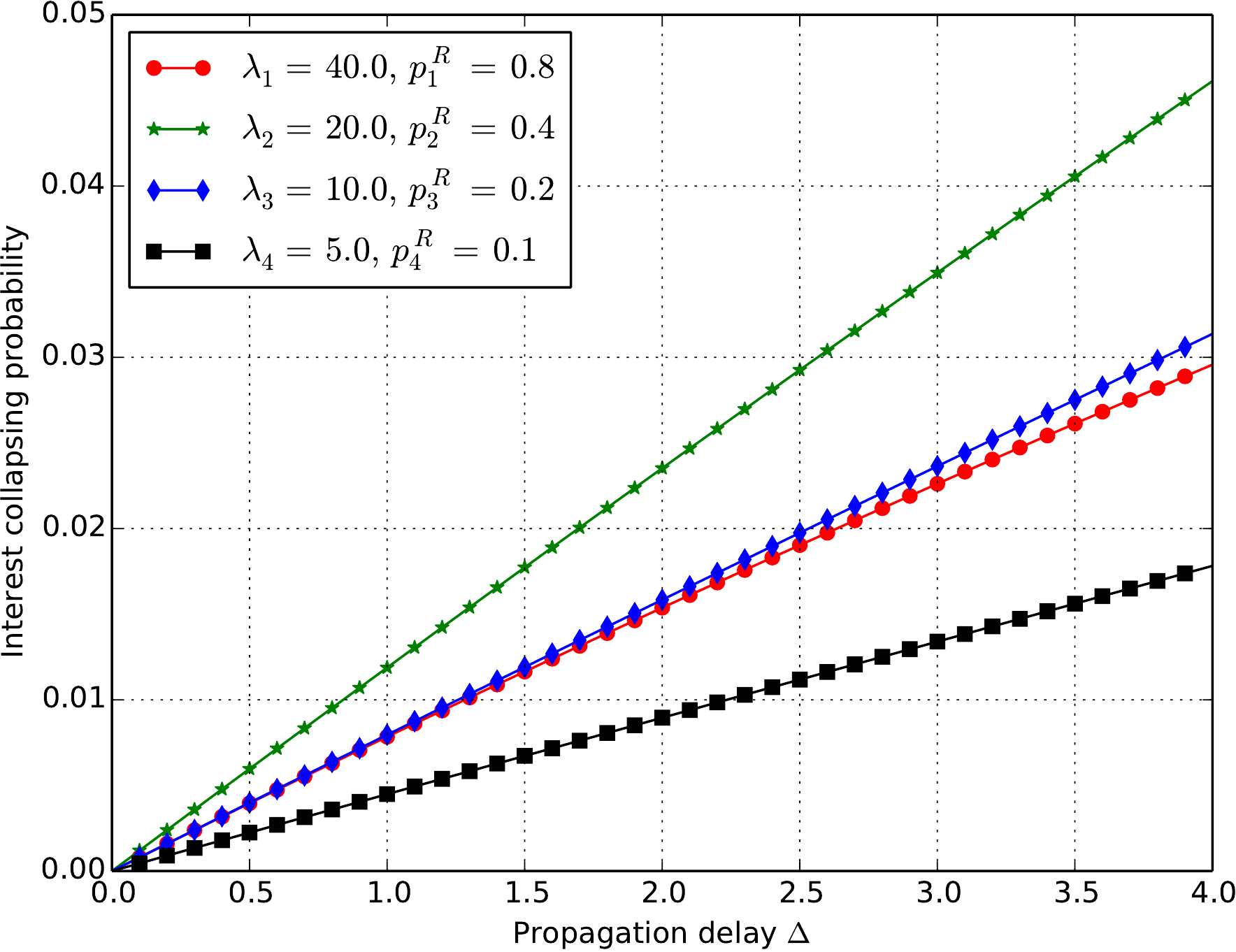}
	\label{fig:int_coll_with_cache}
}
\caption{Interest collapsing probability at $R$}
\label{fig:int_coll}
\end{figure*}

\begin{theorem}
\label{thm:int_coll}
Assuming in-network routing is only enabled at edge routers \cite{garcia2014understanding},
the interest collapsing probability at consumer-facing router $R$ is:
\begin{align}
\label{equ:int_coll_2}
\Pr \left[ \textsf{Coll}_{int}^R(k) \right] =
	\left( 1 - p_k^R \right) \left( 1 - \left( \prod_{i = 1}^{L} e^{- \frac{l_i}{\alpha_i}}
	\right)^{\frac{2 \lambda_k}{c}} \right)
\end{align}
for $L$ links between $R$ and producer $P$, $c=3\times~10^8 m/s$ the speed of light, $l_i$ 
the length of link $i$, constant $\alpha_i$ that depends on the characteristics of the link's 
physical material, and $p_k^R$ the cache hit probability of a class $k$ content at $R$.
\end{theorem}

\begin{proof}
We focus on modeling interest collapsing for individual content
objects. Thus, we set content size $\sigma_k = 1$ segment. Therefore, Equation
\ref{equ:int_coll_1} can be re-written as follows.
\begin{align*}
\Pr \left[ \textsf{Coll}_{int}^R(k) \right] = 1 - e^{-\Delta \lambda_k}
\end{align*}
However, taking into consideration content caching at $R$, the previous equation
can be further re-written as:
\begin{align}
\label{equ:int_coll_3}
\Pr \left[ \textsf{Coll}_{int}^R(k) \right] = \left( 1 - p_k^R \right) \left( 1 - e^{-\Delta \lambda_k} \right)
\end{align}
where $\left(1-p_k^R \right)$ is the cache miss probability. In other words,
if requested content is cached, $R$ satisfies corresponding interests
without creating PIT entries and interest collapsing does not occur.

We now redefine $\Delta$ as a function of propagation delays on each
link on the path: $R\leftrightarrow~P~\leftrightarrow R$.
$\delta_i$ is the propagation delay of the link between $r_{i - 1}$ and $r_{i}$,
where $r_0=R$ and $r_L=P$. Moreover, $p_k^i$ represents cache hit probability
at $r_i$ and, if an interests generates a cache hit at $r_i$, it is not propagated further.
\begin{align*}
\Delta &= 2 \cdot \sum_{i = 1}^{i^*} \left( \delta_i \left( 1 - p_k^i \right) \right) \\
&= 2 \cdot \sum_{i = 1}^{i^*} \left( \frac{l_i}{\alpha_i c} \left( 1 - p_k^i \right) \right)
\end{align*}
where $\alpha_i c$ represents the propagation speed of link $i$,
and $1 < i^* < L$ is the index of router $r_{i^*}$ where a cache hit first occurs.

However, assuming that in-network caching only happens at the edges and
that $\delta_L$ is negligible relative to $\delta_1 + \delta_2 + \dots + \delta_{L - 1}$
(we ignore the effect of caching at $r_{L-1}$), cache hit probability in all
routers between $R$ and $P$ (not including $R$) is zero, and:
\begin{align*}
\Delta = 2 \cdot \sum_{i = 1}^{L} \frac{l_i}{\alpha_i c}
\end{align*}
Therefore,
\begin{align*}
\Pr \left[ \textsf{Coll}_{int}^R(k) \right] &= \left( 1 - p_k^R \right) \left( 1 - e^{- \left( 2 \cdot
		\sum_{i = 1}^{L} \frac{l_i}{\alpha_i c} \right) \cdot \lambda_k} \right) \nonumber \\
&= \left( 1 - p_k^R \right) \left( 1 - \left( e^{- \sum_{i = 1}^{L}
\frac{l_i}{\alpha_i}} \right)^{\frac{2 \lambda_k}{c}} \right) \nonumber \\
&= \left( 1 - p_k^R \right) \left( 1 - \left( \prod_{i = 1}^{L}
e^{- \frac{l_i}{\alpha_i}} \right)^{\frac{2 \lambda_k}{c}} \right)
\end{align*}
This concludes the proof. 
\end{proof}
We analyze Theorem \ref{thm:int_coll} in the following setup.
For simplicity's sake, we use Equation \ref{equ:int_coll_3} for content arrival
rates and propagation delays between $R$ and $P$. Since content popularity
follows a Zipf distribution, arrival rate for class $k+1$ is half of that
for class $k$, i.e., $\lambda_{k + 1} = \lambda_k / 2$. To illustrate the highest
possible interest collapsing probability, we assume that requested content (even if
popular) is not cached at $R$. Figure \ref{fig:int_coll_no_cache} shows the
collapsing probability of four content classes $k = [1,4]$. The
graph only considers propagation delay up to $4$ milliseconds because, as shown
in \cite{carofiglio2011modeling}, the virtual RTT (VRTT)\footnote{RTT taking into
consideration existence of caches.} for content class $k = 4$ is around $4$
milliseconds. We note that $\Pr~\left[ \textsf{Coll}_{int}^R(k)\right]\leq~0.15$
for the most popular content $(k = 1)$. However, in a more realistic setup where $R$'s
cache is taken into consideration, the highest interest collapsing probability is
$<0.05$ for content class $k=2$; see Figure \ref{fig:int_coll_with_cache}.
Based on such low probabilities, we conclude that interest collapsing is not
crucial for a content distribution network such as CCN.

\subsection{Flow and Congestion Control}\label{sec:flow}
Yi et al. \cite{yi2013case} presented the first thorough argument in support of a stateful forwarding
plane in the context of NDN. Due to their near-identical features, the same applies to CCN. 
The PIT can be used to record RTTs for interest and content exchanges, which, in
turn, is useful for making dynamic forwarding decisions. For instance, if the RTT
for a given namespace on a particular link becomes too high, that link
might be congested and alternatives should be explored. This type of in-network
congestion and flow control has been studied further in \cite{yi2012adaptive,rozhnova2014extended,saino2013cctcp}.
For example, \cite{carofiglio2012joint} propose a joint hop-by-hop (i.e., in-network) and 
receiver-based control protocol that relies on PIT-based RTT measurements for flows.

However, according to \cite{wang2013improved}, flow differentiation is a difficult
challenge. Thus, one approach to ``interest shaping'' is by controlling the
flow of data on upstream and downstream links independently of flows. 
This does not require any information from the PIT. Instead, it relies on knowledge
of average interest and content size, link bandwidths, and interest
arrival rates (or demand). Similar to \cite{saino2013cctcp}, it also relies on
receiver-driven flow control via an Additive-Increase-Multiplicative-Decrease 
window. \cite{ren2015interest} is another example of a receiver-driven
flow control protocol for CCN. Given these results, current trends  
in the ICN research community favor pushing stateful control protocols to
receivers, rather than to network nodes.

\subsection{Infrastructure Security}\label{sec:reflection}
Denial of Service (DoS) attacks are a major threat to 
any network infrastructure. DoS attacks in today's Internet include: bandwidth
depletion, DNS cache poisoning, black-holing and prefix hijacking,
as well as reflection attacks. Gasti et al. \cite{gasti2013and} show how CCN (in
the context of NDN) prevents these types of attacks. Out of all attack types considered,
the PIT is needed only to prevent reflection attacks \cite{syverson1994taxonomy}. 
Since content is forwarded based on PIT entries such attacks are impossible in CCN.
However, forwarding content via the PIT is not the only way to prevent
reflection attacks. If packets have a source address, the ingress
filtering technique in \cite{ferguson2000network} -- whereby ISPs filter
packets based on source addresses --  would work equally well.

Despite its resilience to reflection attacks, CCN is susceptible to another
major attack type known as Interest Flooding (IF) \cite{gasti2013and}.
In one IF attack flavor, a malicious consumer (or a distributed botnet) 
issues nonsensical interests\footnote{For example, an interest with a name reflecting a 
valid producer's prefix, with a random number as its last component.} 
so as to overwhelm targeted routers and saturate their PITs. According to 
\cite{carofiglio2015pending}, the PIT size can exceed $10^6$ as upstream paths become congested.
The problem worsens if a malicious consumer and producer cooperate to target a 
specific router. Although several attempts to detect, mitigate, and prevent
them have been made \cite{virgilio2013pit,compagno2013poseidon,you2014detecting, 
afanasyev2013interest,nguyen2015detection,nguyen2015optimal,tang2013identifying}\footnote{For details, see 
Section \ref{sec:related} below.} each of them is effective
against only a very na\"ive or weak attacker. Thus, IF attacks remain a big open 
problem with no comprehensive solution in sight.

%% file: 03-pitless.tex
\section{Stateless CCN using Backwards Routable Names} \label{sec:pitless}
Based on our earlier discussion, the price of a PIT comes at the price of 
serious infrastructure security problems that have not been addressed. To 
this end, we introduce a modified stateless CCN architecture, called stateless CCN.

The main idea behind our stateless CCN design is simple: an interest now includes a new field  
called Backwards Routable Name (BRN), a routable prefix, similar to an IP source address.
A BRN indicates \emph{where} the corresponding content should be delivered, akin to an
IP destination address.
The corresponding content carries the BRN as its routable name towards the consumer.
Thus, with properly configured FIB entries, content is correctly delivered to the
requesting consumer.\footnote{This requires consumers to publicly advertise their BRN prefixes
and participate in routing.} This modification to the CCN architecture
is clearly inspired by IP -- all packets (interest and content) are
forwarded based on addresses they carry and not on network state.
However, as we show below, this does not violate CCN's core value of named
data being moved through, and stored in, the network.

To illustrate BRN-based forwarding, consider a scenario where a consumer
$Cr$ with name \ccnname{lci:/edu/uci/ics/bob} ($N_{Cr}$) requests content from a
producer $P$ with the name \ccnname{lci:/bbc/news/today} ($N_{bbc}$).\footnote{Names
are encoded using the Labeled Content Identifier (LCI) schema \cite{mosko2013ccnx}.
LCI names are the concatenation of individual name components, separated by the
\ccnname{`/'} character, in a typical URI-like format.} Let $I\/nt[N, SN]$ be an
interest with the Routable Name $N = N_{bbc}$ and Supporting Name $SN = N_{Cr}$.
Also, let $C[N, SN]$ be the corresponding content object that matches $I\/nt[N,SN]$
where $C.N = I\/nt.N$, and $C.SN = I\/nt.SN$. In this example, assume that
$C[N,SN]$ is not cached anywhere.
\begin{enumerate}
\item $Cr$ advertises its name $N_{Cr}$ and the routing protocol
propagates this information accordingly.
\item $Cr$ issues $I\/nt[N_{bbc}, N_{Cr}]$.
\item\label{stp:net-frwd} The network forwards $I\/nt[N_{bbc}, N_{Cr}]$ to $P$ according
to router FIB entries.
\item Once $P$ receives $I\/nt[N_{bbc}, N_{Cr}]$ it replies with $C[N_{bbc},N_{Cr}]$.
\item Similarly to Step \ref{stp:net-frwd}, the network forwards $C[N_{bbc},N_{Cr}]$
back to $Cr$ using the same interest forwarding strategy.
\end{enumerate}
Several modifications need to be made to the existing CCN architecture
and protocol to enable this communication.
At a minimum, interest and content object messages should carry two names:
one of the requested content and the other of the requesting consumer.
These two names corresponds to source and destination IP addresses in today's
Internet.

We suggest modifying both interest and content headers to include a new 
field called \texttt{SupportingName} (SN). This field contains the BRN 
of the interest-issuing consumer.
In the above example, interest header would contain \ccnname{lci:/cnn/news/today} and
\ccnname{lci:/edu/uci/ics/bob} as $N$ and $SN$, respectively. The replied content header
would contain the same $N$ and $SN$ values. Note that content object signatures
can be generated in advance by omitting the content's $SN$ field since this is only used for
routing purposes. The resulting packet formats are shown in Figure \ref{fig:packets} in ABNF form.

\begin{figure}
\small\bfseries
\begin{Verbatim}[frame=single,fontsize=\scriptsize]
Message := MessageType PacketName [Payload] [Validation]
MessageType := Interest | ContentObject
PacketName := Name SupportingName
Name := CCNx Name
SupportingName := CCNx Name
Payload := OCTET+
Validation := ValidationAlg ValidationPayload
\end{Verbatim}
\caption{Stateless Packet Format in ABNF; {\tt ValidationAlg} and
{\tt ValidationPayload} elements are defined in \cite{messages}.}
\label{fig:packets}
\end{figure}

Currently, interest and content messages are very similar in CCN.
Both contain a Name, Payload, and optional Validation fields \cite{messages}; they
only differ in the top-level type. Our stateless variant still requires
this distinction since interests and content objects are processed differently.
For example, a router first attempts to satisfy an interest from its
cache, while content is (optionally) cached prior forwarding.

We stress that a content might not follow the reverse path of the proceeding 
interest due to routing table configurations. In fact, we anticipate that consumers 
might structure BRNs to control the degree of path asymmetry between 
interest and content messages.

Modified interst and content formats coupled with PIT removal 
simplify router's fast-path processing. Algorithms \ref{alg:int} and \ref{alg:data} 
show how a router would process interest and content messages. {\sf CS-Lookup}
represents a CS lookup operation based on $N$ (content name).
For clarity's sake, we omit content verification details in all algorithms.

\begin{algorithm}[ht!]
\caption{{\sf Process-Interest}} \label{alg:int}
\begin{algorithmic}[1]
\small
\State {\bf Input:} Interest $I\/nt[N,SN]$, arrival interface $F_i$, CS, FIB
\State $C = $ {\sf CS-Lookup}$(\mathrm{CS},N)$
\If {$C \not= \mathrm{nil}$}
    \State {\bf Forward} $C$ to $F_i$
\Else
    \State $F_o = $ FIB.{\sf Lookup}$(N)$
    \State {\bf Forward} $I\/nt[N,SN]$ to $F_o$ based on local strategy
\EndIf
\end{algorithmic}
\end{algorithm}

\begin{algorithm}[ht!]
\caption{{\sf Process-Content-Object}} \label{alg:data}
\begin{algorithmic}[1]
\small
\State {\bf Input:} Content Object $C[N,SN]$, CS, FIB
\State Cache $C[N,SN]$ with $N$ as the key
\State $F_o = $ FIB.{\sf Lookup}$(SN)$
\State {\bf Forward} $C[N,SN]$ to $F_o$ based on local strategy
\end{algorithmic}
\end{algorithm}

%% file: 04-assessment.tex
\section{Architecture Assessment}\label{sec:assessment}
Despite significant research progress over the past 5 years, the PIT no
longer seems to be a practical solution for content object forwarding in CCN.
As discussed earlier, router PITs are
prone to DoS (specifically, IF) attacks. They also store information
already available from FIBs (consumer routable prefixes) and enforce unnatural
path symmetry in an increasingly asymmetric Internet.
Moreover, flow and congestion control algorithms are being pushed towards receivers
instead of in the network based on state maintained in PITs. 
Our simple stateless CCN variant mitigates these problems by specifying the
use of source and destination prefixes. To support our claims, we
compare the stateful and stateless CCN architectures
with respect to aforementioned features. We then discuss both advantages
and disadvantages of stateless CCN.

\subsection{Revisiting PIT Benefits}\label{sec:pit-benefits}
\noindent
{\bf Reverse-Path Routing.}
Our stateless CCN scheme requires FIBs to be updated to accommodate
RBN prefixes advertised by consumers. It might seem that this would
lead to a tremendous increase in FIB size. However, recall that CCN
interest (and now, content) forwarding is based on LPM. In stateless CCN,
consumers announce their RBNs only to their
first-hop routers (e.g., an access point), which, in turn, combines all
its consumers' RBNs and announces an aggregate prefix to neighboring
routers, similar to the Border Gateway Protocol (BGP)
route-aggregation feature \cite{stewart1998bgp4}.

Also, path asymmetry between interest and content messages
in stateless CCN is more compliant with networking and routing practices of
today's Internet. As argued in Section \ref{sec:content-forwarding}, ISPs are
likely to adopt an architecture that agrees with their present business model.

\noindent
{\bf Forwarding Overhead.}
Stateful CCN dictates that, when processing an interest, a router should,
in the worst case: (1) attempt to satisfy the interest from its cache, (2) create or
modify a PIT entry for the interest, and (3) perform a FIB lookup.
Meanwhile, stateless CCN eliminates (2), which reduces router
operations to cache and FIB lookups.

Removing the PIT also simplifies and improves content forwarding
logic. Instead of indexing the PIT to obtain downstream interfaces, a router performs 
a FIB lookup for each content, just as it does in processing an interest. We claim
that a simple FIB lookup is more efficient than using the PIT for forwarding content.
After a PIT lookup, a PIT entry is flushed once a content object is forwarded. This
"flush" incurs an additional write operation in the routers' fast path. Stateless CCN
replaces this with a simple LPM-based FIB lookup. Note that LPM algorithms have 
been intensively studied, constructed, and fine-tuned to cope with 
multi-gigabit, and even terabit, IP packet
processing \cite{waldvogel2007fast,kobayashi2000longest,dharmapurikar2003longest}.

\noindent
{\bf Flow and Congestion Control.} 
The current receiver driven flow and congestion control algorithms are 
unaffected in our stateless variant. The only difference is that now routers
are unable to compute the RTT for a given interest-content exchange. 
Given that recent algorithms do not rely on these calculations \emph{in
the network} anyway, this is a tolerable loss.

\subsection{Content Caching}\label{sec:caching}
As mentioned earlier, using RBNs for content routing does not preserve
path symmetry. In fact, it encourages path asymmetry. Consequently,
content might be cached along a different path than the interest
originally traversed. It might seem that adjacent (or nearby) consumers
for the same content would therefore not benefit from in-network caching.
We argue that this is not so.

Firstly, CCN content includes a producer-specified cache hint that suggests 
how long routers ought to cache this content. Routers are expected to honor
this hint when managing the data in their caches. Secondly, recall that 
routers can unilaterally decide whether to enable caching for none, some, 
or all content. Thirdly, cache eviction strategies might reduce cache entry lifetime
to much less than the suggested value. For instance, core routers would most
likely not cache content given their high processing rates.
Meanwhile, consumer-facing routers would handle much less traffic
and are thus more likely to cache content for the required amount of time. In fact,
caching has been shown to be most cost effective at
the edges \cite{garcia2014understanding}, e.g., at the tier-3 ISP level.
Since nearby consumers share the same edge router, they will
all benefit from caching popular content in that router.
This observation is supported by the results obtained in
\cite{john2010estimating}, wherein it is shown that path symmetry is highest
at the edges of the network.

\begin{figure}[t]
\centering
\includegraphics[width=0.7\columnwidth]{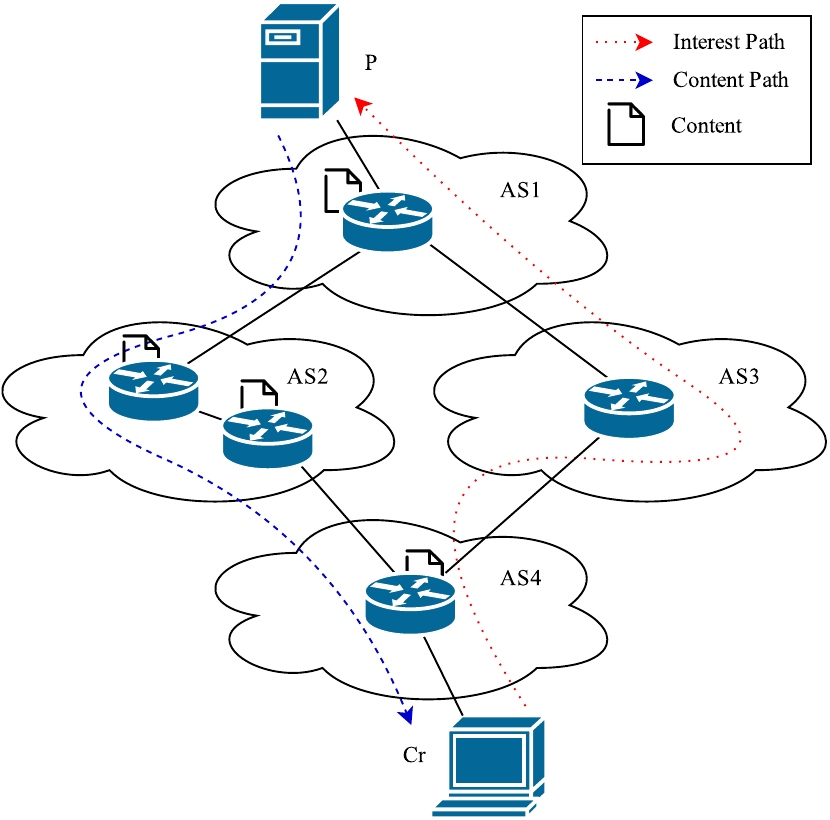}
\caption{Caching in stateless CCN. AS1 and AS4 are stub autonomous system
representing tier-3 ISPs, AS2 and AS3 are transit autonomous system representing
tier-1 ISPs.}
\label{fig:caching}
\end{figure}

Figure \ref{fig:caching} shows an example of caching in stateless CCN. The
topology has 4 autonomous systems (AS-s). AS1 and AS4 are stubs
representing tier-3 ISPs, while AS2 and AS3 are transits representing tier-1
ISPs.\footnote{We ignore tier-2 ISPs for simplicity.}
Interests issued by $Cr$ are forwarded towards $P$
along the dotted (red) path, and content is forwarded back to $Cr$ along the
dashed (blue) path. Assuming that caching only occurs near the edges,
content sent from $P$ to $Cr$ gets cached in AS4. Consequently,
interests for the same content issued by other consumers in AS4 would
be satisfied from cache(s) of AS4.

\begin{figure*}[t]
\begin{center}
\subfigure[DFN topology.]
{
	\includegraphics[width=0.7\columnwidth]{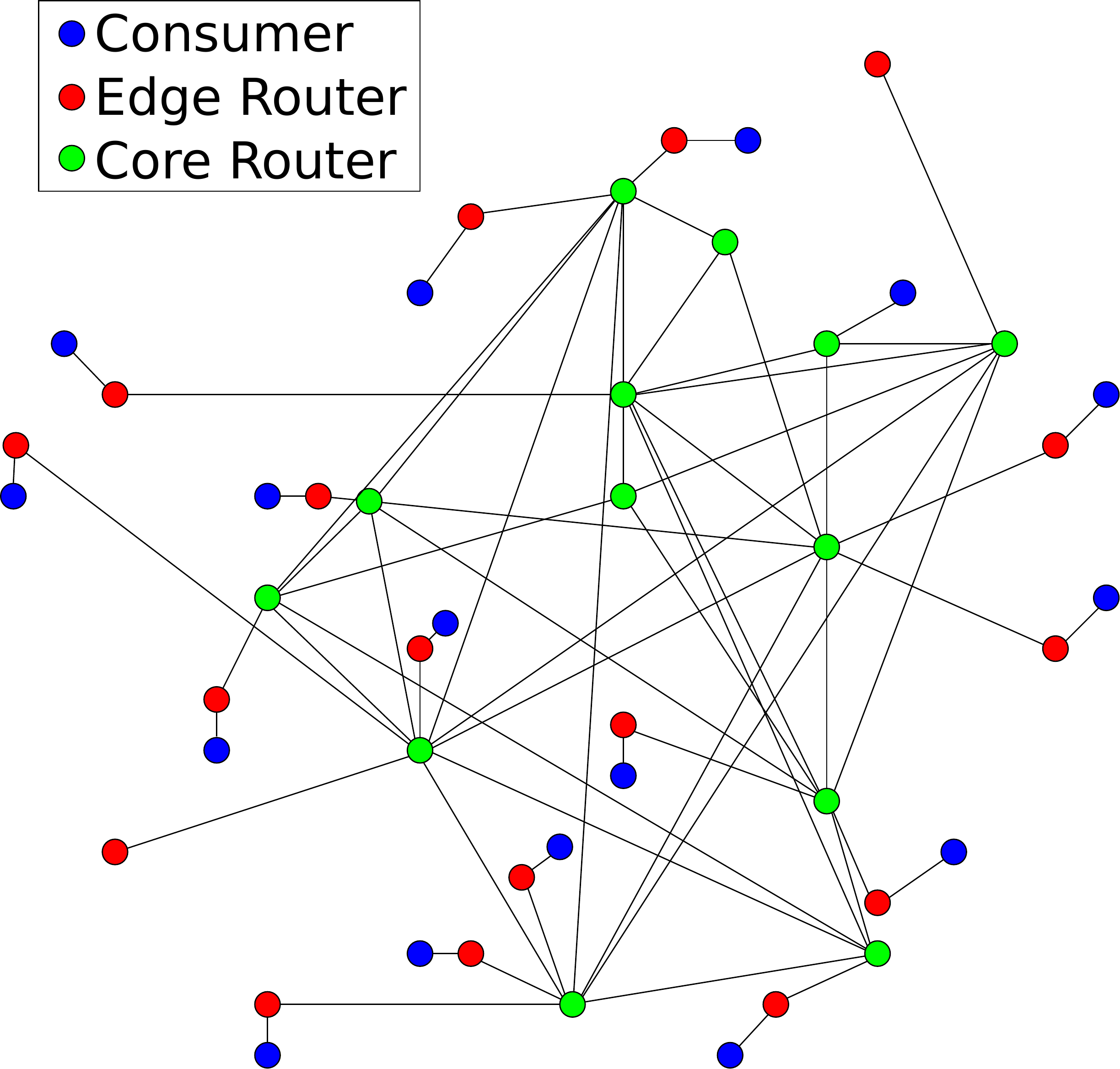}
	\label{fig:dfn-topology}
}
\subfigure[AT\&T topology.]
{
	\includegraphics[width=\columnwidth]{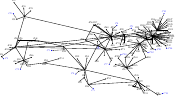}
	\label{fig:att-topology}
}
\caption{The DFN and AT\&T topologies.}
\label{fig:tops}
\end{center}
\end{figure*}

\subsection{Infrastructure Security}\label{sec:infra-sec}
We now discuss both beneficial and problematic infrastructure security 
issues in stateless CCN, such as (D)DoS attacks.

\noindent
{\bf Interest Flooding.} Stateless CCN mitigates
this attack by eliminating its root cause -- the PIT. Without per-request
state in routers, this attack vector is removed. This represents the major 
advantage of stateless CCN.

\noindent
{\bf Reflection Attacks.} Interest and content path symmetry in CCN prevents
reflection attacks. However, in stateless CCN, RBNs serve as a {\em de facto}
source address in interest, and destination in content, messages.
Thus, reflection attacks re-appear. Fortunately, the ingress filtering 
technique described in \cite{ferguson2000network} can be used to
mitigate them.

\noindent
{\bf Cache and Content Poisoning}. Content authentication in stateless CCN is identical
to that in the stateful CCN architecture. It is done by producers
signing content objects or using Self-Cerftifying Names (SCN)
\cite{ghali2014network}. Regardless of the method, all content {\em must} be
verified by consumers. However, verification is not mandatory 
for routers, for several reasons; see \cite{ghali2014network} for more
details. Lack of in-network content verification opens the door for content poisoning
attacks \cite{ghali2014needle}. Moreover, due to possible path asymmetry
in RBN-based content forwarding, content poisoning countermeasures
that work in the current CCN architecture do not apply anymore.

The PIT enables a router to apply the so-called Interest-Key Binding
(IKB) rule \cite{ghali2014network}, whereby consumers and
producers collaborate to provide routers with enough (minimal) trust information to
perform content verification. This information is currently stored in
the PIT. However, as mentioned above, path asymmetry renders
IKB impractical. In stateless CCN, a router might receive (unsolicited) content 
without prior interest traversing the same path. 
If such content is returned on a path different from the
original interest, routers cannot trust any information it carries.

\noindent
{\bf Consumer Privacy}. Lack of source addresses in stateful CCN facilitates
a degree of consumer privacy. RBNs in stateless CCN negate this benefit.
However, given highly descriptive nature of content names, end-to-end encryption
might be needed to achieve better privacy, even in stateful CCN. Moreover, since
many current applications customize content per consumer, both interest and
content messages, if left unencrypted, might include enough information to identify the consumer.
Thus, for sensitive or consumer-specific content, stateless CCN does not necessarily
yield worse privacy due to the use of encryption.

Beyond naming issues, in-network caching can be abused by an adversary to compromise both consumer
and producer privacy \cite{acs2013cache}. By measuring the time required
for content retrieval, an adversary can learn whether specific content was
recently requested by other nearby consumers. This attack is still applicable
in stateless CCN since in-network caching remains a feature. Fortunately,
countermeasures proposed in \cite{acs2013cache,mohaisen2013protecting} are equally
effective in both architectures.

\subsection{Deployment Issues}\label{sec:deployment}
The main purpose of designing a stateless CCN architecture is to provide an alternative
to the current stateful CCN. This does not mean that one must replace the other.
In fact, they can co-exist and allow the consumers to
select one or the other on-demand. Consider the following scenarios:
\begin{compactenum}
\item Stateless CCN: $Cr$ includes an RBN ($SN$)
in an interest and upstream routers forward it as necessary. Stateful
routers create PIT entries and stateless routers do not.
In both cases, the interest is forwarded according to the FIB using
content name $N$. Upon receipt of a content message, a stateful router
uses its PIT to forward the content downstream, while a stateless router does
that using the FIB and $SN$. In this case, stateful forwarders
simply ignore the $SN$ fields in both interests and content objects. This
makes the proposed stateless CCN backwards compatible with
the current CCN architecture.
\item Stateful CCN: $Cr$ issues an interest per current CCN rules.
If a stateless router receives such an interest,
it generates a NACK indicating that the interest cannot be
forwarded further. To handle this NACK, some downstream node must provide a
RBN for the interest and re-forward it as needed.
This node can be an AS gateway (i.e., a router that can forward packets
to and from other ASs) or, worst case, the consumer.
\end{compactenum}
Any node that satisfies an interest must honor its version (stateless or
stateful) when producing a response. For example, if a producer
(or a caching router) receives an interest
with an RBN, it must reply according to stateless CCN by keeping
both $N$ and $SN$ in the corresponding content.

We claim that this type of hybrid approach aligns very well with CCN's edge-caching
strategy \cite{garcia2014understanding} and current path asymmetry in the Internet's core.
Recall that stateless CCN makes it impossible for stateless routers to verify
content they forward. Moreover, if a router cannot verify a content
then it should not cache it. Thus, caching would only occur near the
edges where PITs are located. We envision a network where
all routers within an AS have PITs that are used for \emph{egress interests}, i.e.,
those generated by consumers within the AS. Traditional verification techniques
can be applied at these routers for returned content. If an interest
leaves an AS, the gateway first supplies a RBN before forwarding it
upstream. This interest will not induce any PIT state upstream and therefore
will not result in the corresponding content object being cached outside of the AS.
This is not problematic though since the results of \cite{garcia2014understanding}
imply that caching near the edge (i.e., within the AS above) is most effective.
Moreover, this approach allows path asymmetry outside of the AS, which aligns
with the real-world routing strategies noted in \cite{john2010estimating}.

%
Another side-effect of the hybrid approach is that it can be used as an IF attack recovery
mechanism. If $R$ implements a PIT but does not
have enough resources to create a new entry for $I\/nt$, $R$ can respond with a NACK
similar to what is described above. In this case, if $R$ is under an IF attacks and
its PIT resources is exhausted, neither current nor future interests will be dropped.
The disadvantage of this approach as an effective IF attack countermeasure is that
(1) it is reactive, so it can only be used after the attack occurs, and (2) it
incurs an additional end-to-end latency since consumers (or downstream routers)
need to reissue the interests following stateless CCN guidelines.

%% file: 05-experiments.tex

\begin{figure*}[t]
\begin{center}
\subfigure[Interest processing overhead in the DFN topology with $160$ consumers.]
{
	\includegraphics[width=\columnwidth]{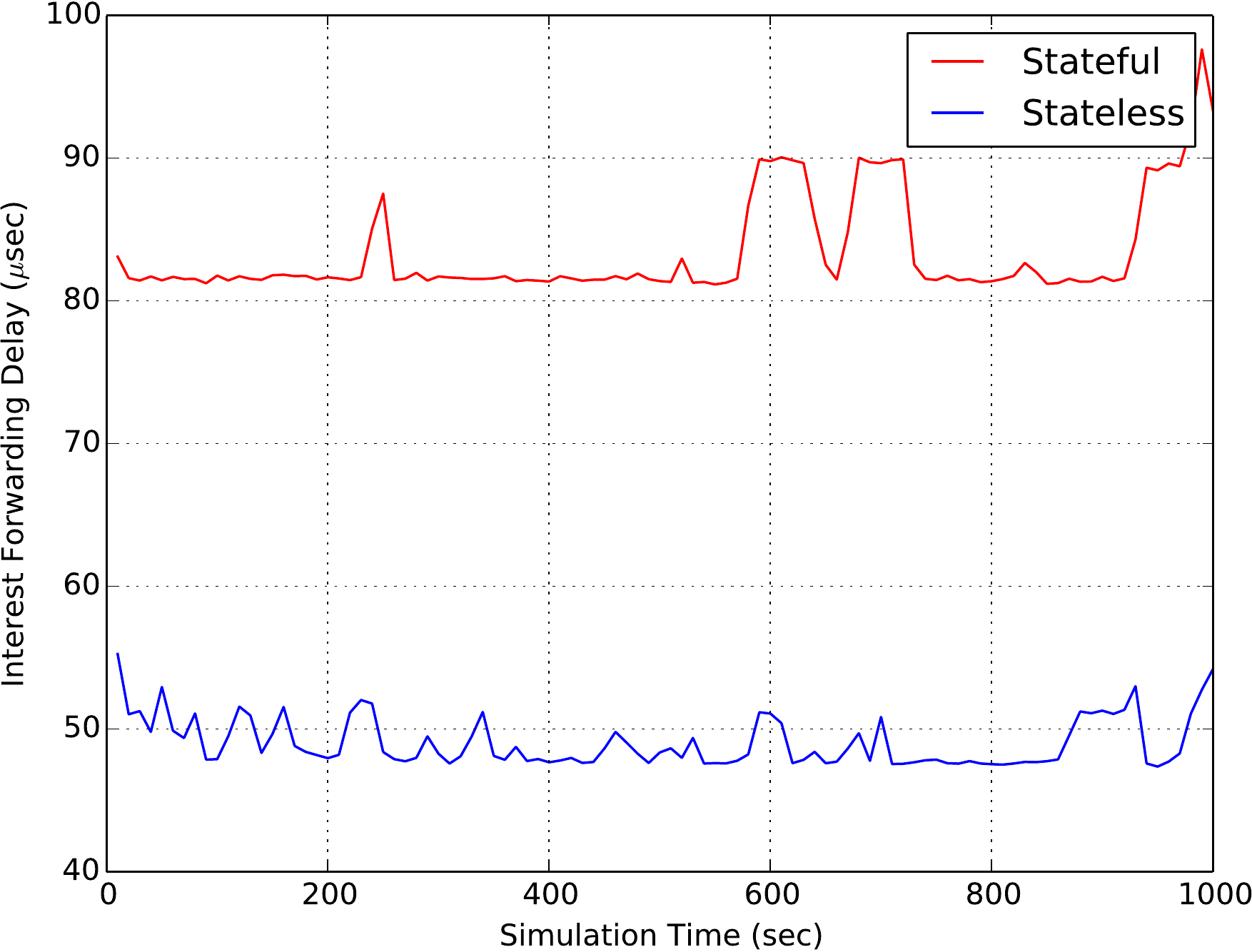}
	\label{fig:att-2}
}
\subfigure[Data processing overhead in the DFN topology with $160$ consumers.]
{
	\includegraphics[width=\columnwidth]{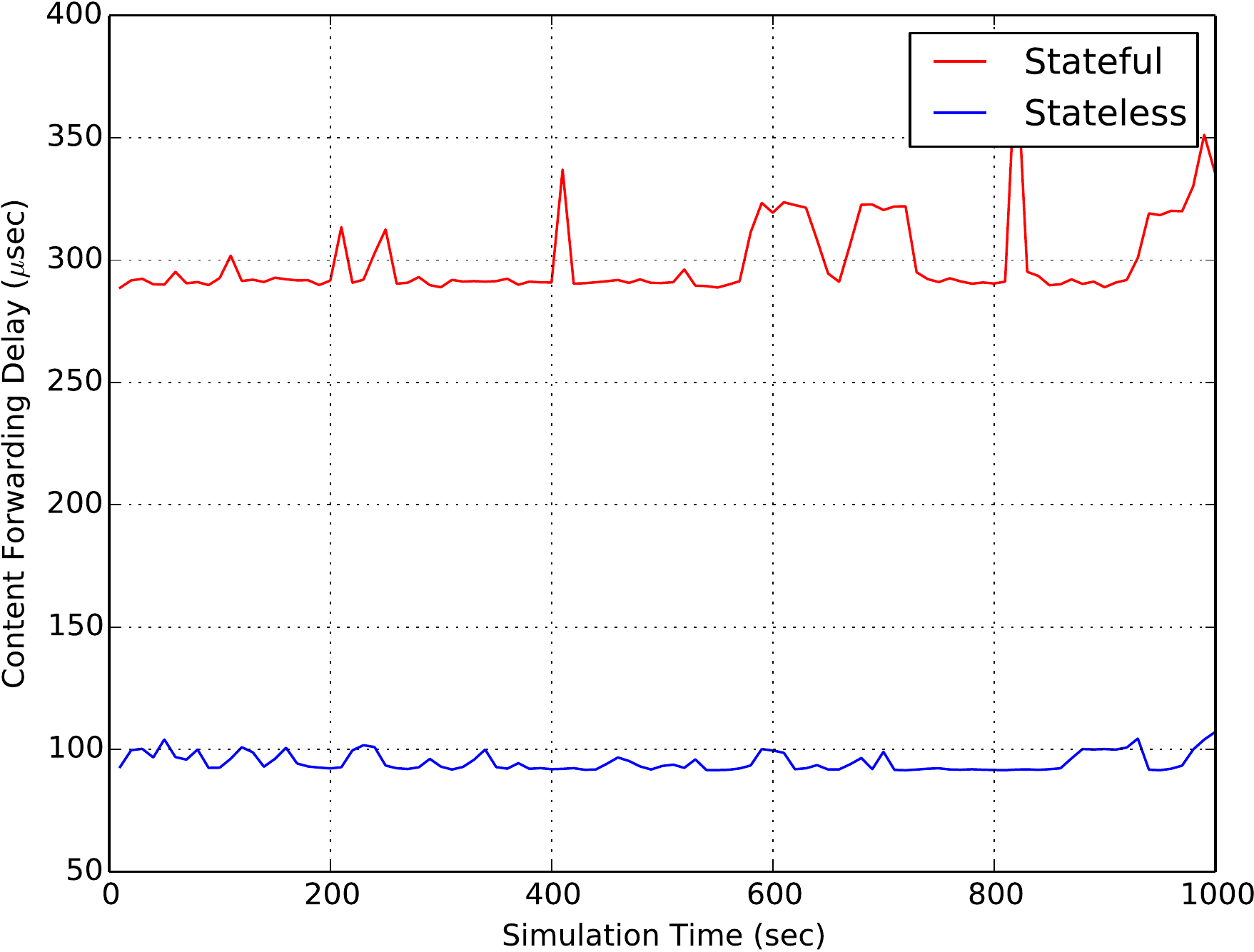}
	\label{fig:att-1}
}
\subfigure[Interest processing overhead in the AT\&T topology with $160$ consumers.]
{
	\includegraphics[width=\columnwidth]{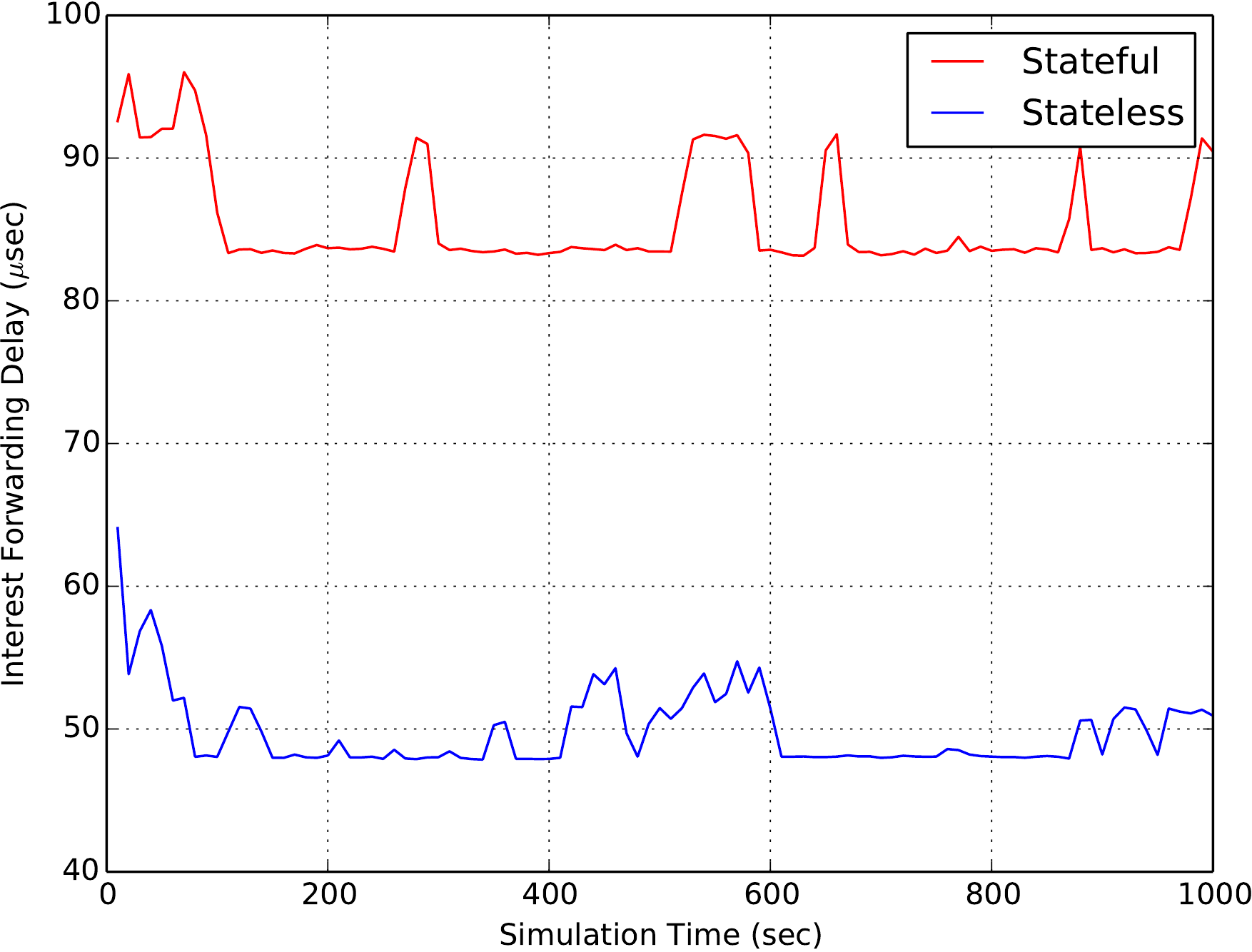}
	\label{fig:dfn-2}
}
\subfigure[Data processing overhead in the AT\&T topology with $160$ consumers.]
{
	\includegraphics[width=\columnwidth]{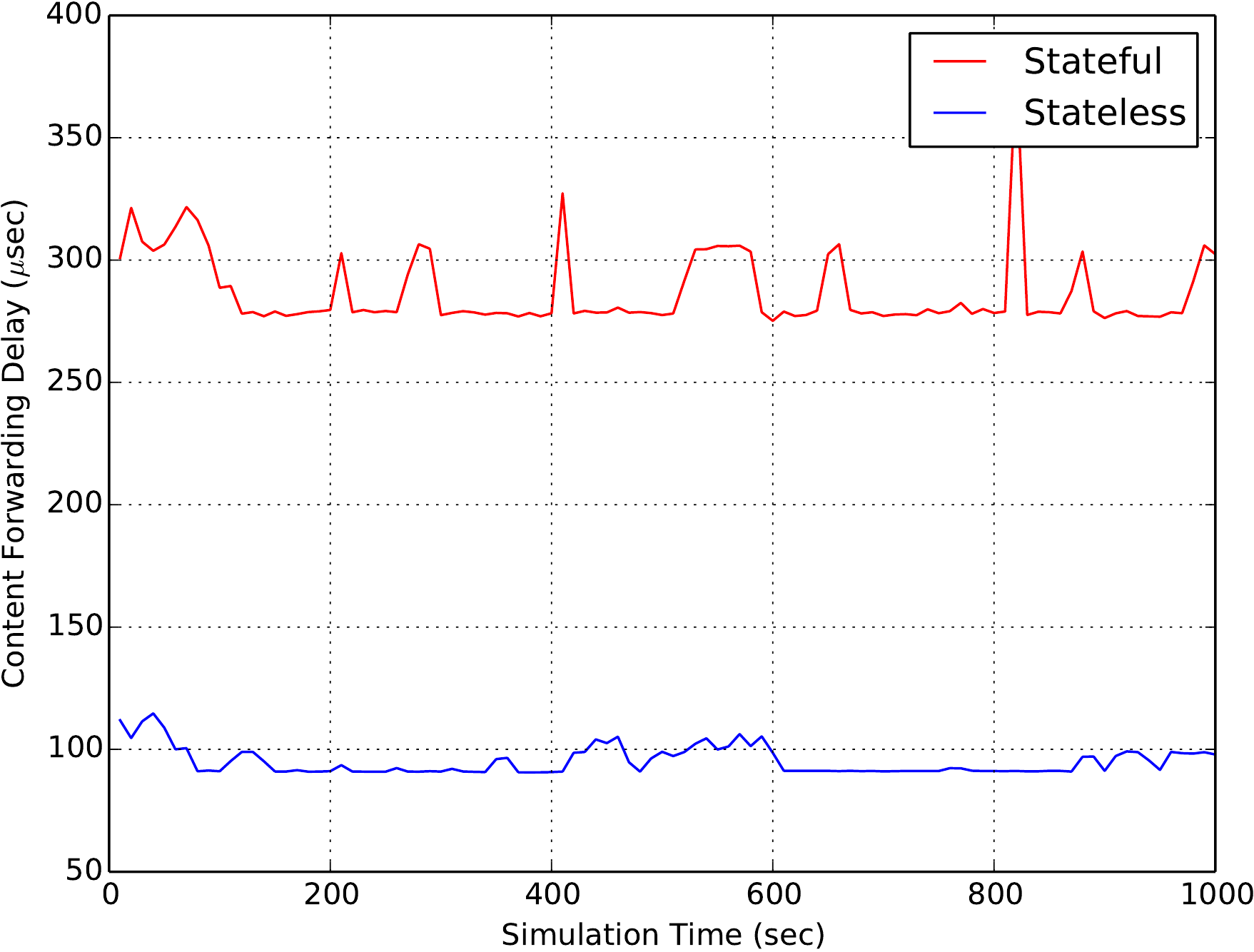}
	\label{fig:dfn-1}
}
\caption{Forwarding overhead in stateful (red) and stateless (blue) CCN variants.}
\label{fig:fwd-results}
\end{center}
\end{figure*}

\begin{figure}[t]
\centering
\includegraphics[width=\columnwidth]{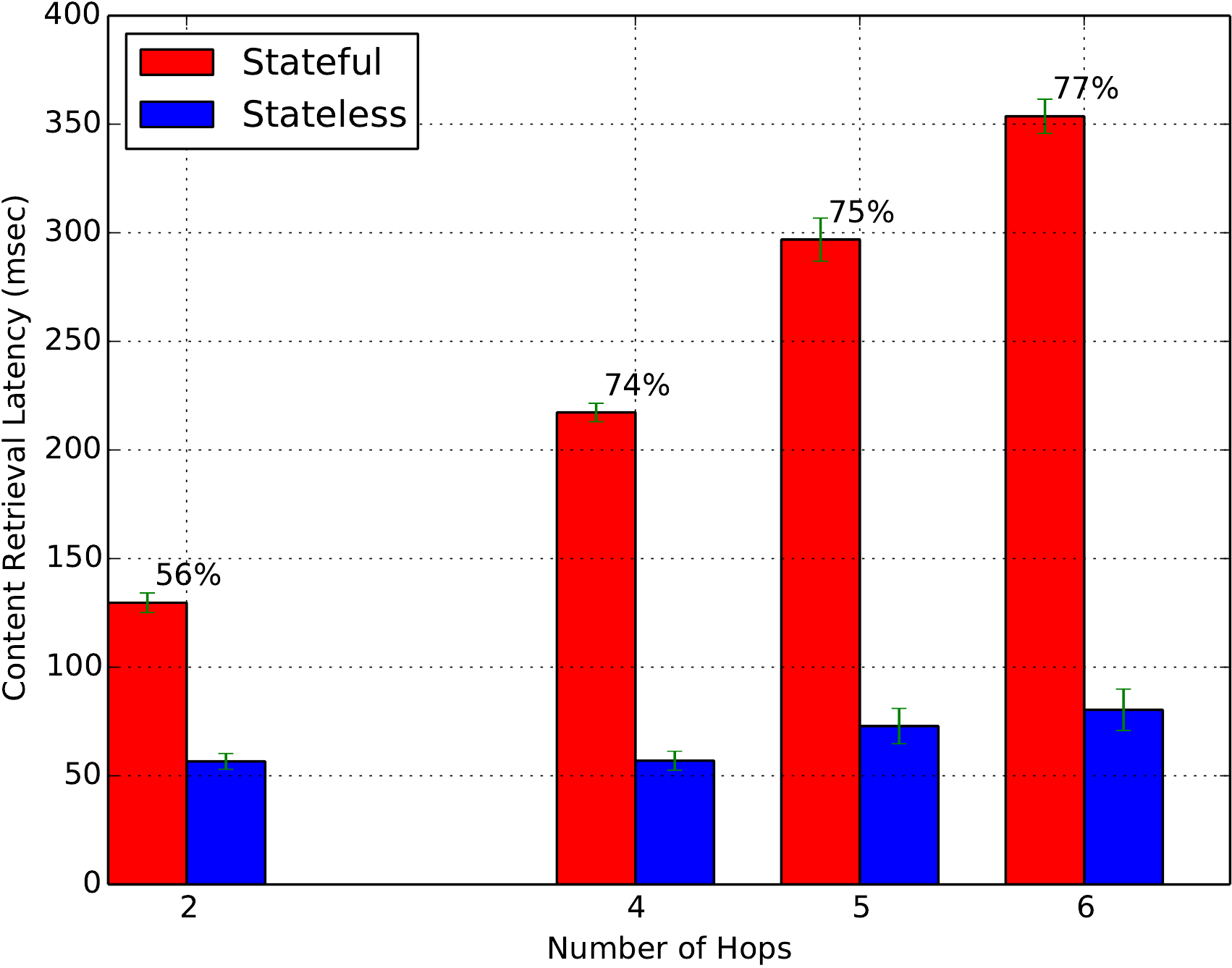}
\caption{Content retrieval latency as a function of number of hops between
consumers and producers for both stateful and stateless forwarders. Note that
paths with 3 hops do not exist in this topology.}
\label{fig:delay-rtt-dfn}
\end{figure}

\section{Experiments and Analysis} \label{sec:experiment}
We now evaluate performance of the stateless CCN in relation
to stateful CCN. Our key metric is the degree to which forwarding overhead
is lowered by stateless routing. To do this, we modified the ndnSIM \cite{ndnSIM} simulator,
a simplified NDN implementation as a NS-3 \cite{ns3} module, to support the stateless
CCN architecture proposed in Section \ref{sec:pitless}. Specifically, we modified the
NDN Forwarding Daemon (NFD) \cite{afanasyev2014nfd} to supporting interest forwarding
based on content names and content forwarding based on RBNs, without leaving PIT states
behind.

We then simulated topologies based on the Deutsches ForschungsNetz (DFN),
the German Research Network \cite{DFNverein, DFN-NOC}, and AT\&T networks (shown
in Figures \ref{fig:dfn-topology} and \ref{fig:att-topology}, respectively).
Each topology consists of $160$ consumers\footnote{Each consumer node in the figures
consists of $10$ actual consumers.}, a single producer connected to one of the edge
routers, and multiple routers (more than $30$).
Each consumer generates $10$ interests per second, with a random suffix so as to
avoid cache hits. This is done to force interest to traverse all the path to the
producer, hence maximize the amount of processing that takes place in the forwarders in
both the upstream and downstream paths. This captures the worst-case scenario.

Our results are shown in Figure \ref{fig:fwd-results}. In both topologies, we
observe approximately $63\%$ improvement in \emph{per packet} forwarding performance
of interest messages. We also observe a $66\%$ improvement in processing content objects.
These cost savings are quite significant, especially, for core routers that might process
packets at rates of $100$Gbps and over.

Furthermore, the overall content retrieval latency is significantly improved
when using stateless forwarders as compared to stateful ones. Figure
\ref{fig:delay-rtt-dfn} shows a comparison of the RTT performance
for both forwarders in the DFN topology. In this experiment, consumers always request
unique content in order to avoid cache hits.\footnote{We do not take caching into
consideration to eliminate any randomize effect it might have on content retrieval
latency. Such randomized effects can be cause by different cache eviction policies.}.
Figure \ref{fig:delay-rtt-dfn} illustrates the average RTT to fetch content
objects as a function of the number of hops between consumers and the producer.
We notice a content retrieval latency improvement of more than $50\%$. This improvement
reaches $77\%$ for paths consisting of $6$ hops.

To summarize, our results show that stateless forwarding leads to better routing 
and content retrieval performance.

%% file: 06-related.tex
\section{Related Work} \label{sec:related}
PIT-focused DoS attacks in CCN are a well-known problem \cite{virgilio2013pit}.
Rate-based \cite{compagno2013poseidon,you2014detecting,afanasyev2013interest}
and statistical-based tests \cite{nguyen2015detection,nguyen2015optimal,tang2013identifying}
have been proposed to detect these attacks and subsequently limit the incoming
interfaces upon which malicious interests arrive. However, this only treats a
symptom of the problem--it does not solve the core issue of PIT state in routers.
Dai et al. \cite{dai2013mitigate} propose a technique called ``interest tracebacks''
to identify malicious attackers and limit the rate at which they
can send messages to the network. The key observation is that PIT state leaves
a trace that terminates at the source of an interest. The network can use this
trail to then identify the attacker. However, this approach depends on localized
attackers sending interests at a high rate; it does not work for highly
distributed adversaries. Similar in-network throttling techniques were discussed in
\cite{compagno2013poseidon} and \cite{afanasyev2013interest}.
Complementary to this general technique, Al-Sheikh et al. \cite{al2015revisiting}
introduce FIB exclude filters that seek to prevent malicious interests from propagating
upstream to locations in the network where the target content cannot possibly be served.
These filters work for static content, only, and cannot be used to prevent
interests for dynamic content from being forwarded.
Li et al. \cite{li2014interest} propose the use of consumer-based puzzles that
must be solved as a native rate-limiting technique. These puzzles, or ``interest
cash,'' are generated by producers to be solved and must be completed \emph{for
each interest}. Although this approach is effective, it severely harms benign consumers.

Techniques to outright replace the PIT have also been proposed. \cite{tsilopoulos2014reducing}
devised a ``semi-stateful'' solution wherein packets are marked (with Bloom
Filters \cite{bloom1970space}) to be forwarded correctly. This approach shifts the state
that was once in the PIT to the packets themselves and creates unnecessary
communication and control overhead in the network.
In a similar vein, Wang et al. \cite{wang2013decoupling} describe a protocol variant wherein
resource-constrained PITs can offload the per-request state \emph{into} interests
that are forwarded. This technique puts PIT state ``on the wire'' and allows
a PIT to naturally decrease in size as content is returned without dropping
interests from benign consumers and routers. This is in contrast to our work
where we put defer state information to the routing protocol.

Salah et al. \cite{salah2015lightweight} used a router coordination
framework called CoMon (Coordination with Lightweight Monitoring) to enable
adjacent nodes to share information about forwarding state and traffic. Select
routers are assigned the role of ``monitor.'' Their goal is to monitor interest
and content exchanges and measure the (un)satisfaction rate. This information is
periodically reported to a central ``domain controller'' that is in charge of
processing the traffic reports to detect and respond to IF attacks. Monitoring
routers are chosen based on their location in the network and closeness to
producers. This solution assumes an unrealistic static topology and centralized
post facto detection mechanism. In summary, this scheme is an extension to previous
rate-based throttling solutions.

Almirshari et al. \cite{almishari2013optimizing} proposed a technique to ``piggyback'' interest
and content objects to enable high throughput bidirectional communication in NDN. Their approach
introduces a new packet type in addition to interests and content objects. It also requires
that interests are unnaturally extended to carry application data in the name. Moreover, their
approach is still susceptible to IF attacks since it requires PIT state for bidirectional communication.
Dai et al. \cite{dai2012pending} study extensions to the PIT to support modern applications such as
streaming services and online gaming. They propose to create long-lived PIT entries to enable
bidirectional communication between clients and servers. This only serves to make adversary's
job easier in an IF attack.

%% file: ccn-pitless.bbl
\begin{thebibliography}{10}
\providecommand{\url}[1]{#1}
\csname url@samestyle\endcsname
\providecommand{\newblock}{\relax}
\providecommand{\bibinfo}[2]{#2}
\providecommand{\BIBentrySTDinterwordspacing}{\spaceskip=0pt\relax}
\providecommand{\BIBentryALTinterwordstretchfactor}{4}
\providecommand{\BIBentryALTinterwordspacing}{\spaceskip=\fontdimen2\font plus
\BIBentryALTinterwordstretchfactor\fontdimen3\font minus
  \fontdimen4\font\relax}
\providecommand{\BIBforeignlanguage}[2]{{%
\expandafter\ifx\csname l@#1\endcsname\relax
\typeout{** WARNING: IEEEtran.bst: No hyphenation pattern has been}%
\typeout{** loaded for the language `#1'. Using the pattern for}%
\typeout{** the default language instead.}%
\else
\language=\csname l@#1\endcsname
\fi
#2}}
\providecommand{\BIBdecl}{\relax}
\BIBdecl

\bibitem{ahlgren2012survey}
B.~Ahlgren \emph{et~al.}, ``A survey of information-centric networking,''
  \emph{IEEE Communications}, vol.~50, no.~7, 2012.

\bibitem{jacobson2009networking}
V.~Jacobson \emph{et~al.}, ``Networking named content,'' in \emph{CoNEXT},
  2009.

\bibitem{ccn10}
I.~Solis, ``{CCN} 1.0 (tutorial),'' in \emph{ICN}, 2014.

\bibitem{yi2012adaptive}
C.~Yi \emph{et~al.}, ``Adaptive forwarding in named data networking,''
  \emph{ACM CCR}, vol.~42, no.~3, 2012.

\bibitem{zhou2015proactive}
J.~Zhou \emph{et~al.}, ``A proactive transport mechanism with explicit
  congestion notification for {NDN},'' in \emph{ICC}, 2015.

\bibitem{park2014popularity}
H.~Park \emph{et~al.}, ``Popularity-based congestion control in named data
  networking,'' in \emph{ICUFN}, 2014.

\bibitem{wang2013improved}
Y.~Wang \emph{et~al.}, ``An improved hop-by-hop interest shaper for congestion
  control in named data networking,'' \emph{ACM CCR}, vol.~43, no.~4, 2013.

\bibitem{carofiglio2013multipath}
G.~Carofiglio \emph{et~al.}, ``Multipath congestion control in content-centric
  networks,'' in \emph{INFOCOM WKSHPS}, 2013.

\bibitem{braun2013empirical}
S.~Braun \emph{et~al.}, ``An empirical study of receiver-based aimd
  flow-control strategies for {CCN},'' in \emph{ICCCN}, 2013.

\bibitem{saino2013cctcp}
L.~Saino \emph{et~al.}, ``{CCTCP}: A scalable receiver-driven congestion
  control protocol for content centric networking,'' in \emph{ICC}, 2013.

\bibitem{gasti2013and}
P.~Gasti \emph{et~al.}, ``{DoS} and {DDoS} in named data networking,'' in
  \emph{ICCCN}, 2013.

\bibitem{so2012toward}
W.~So \emph{et~al.}, ``Toward fast {NDN} software forwarding lookup engine
  based on hash tables,'' in \emph{ANCS}, 2012.

\bibitem{tsilopoulos2014reducing}
C.~Tsilopoulos \emph{et~al.}, ``Reducing forwarding state in content-centric
  networks with semi-stateless forwarding,'' in \emph{INFOCOM}, 2014.

\bibitem{yuan2014scalable}
H.~Yuan \emph{et~al.}, ``Scalable pending interest table design: From
  principles to practice,'' in \emph{INFOCOM}, 2014.

\bibitem{carofiglio2015pending}
G.~Carofiglio \emph{et~al.}, ``Pending interest table sizing in named data
  networking,'' in \emph{ICN}, 2015.

\bibitem{john2010estimating}
W.~John \emph{et~al.}, ``Estimating routing symmetry on single links by passive
  flow measurements,'' in \emph{IWCMC}, 2010.

\bibitem{fette2011websocket}
I.~Fette \emph{et~al.}, ``{RFC} 6455: The websocket protocol,'' 2011.

\bibitem{skype}
``Skype,'' \url{http://www.skype.com/}.

\bibitem{cohen2008bittorrent}
B.~Cohen, ``The {BitTorrent} protocol specification,'' 2008.

\bibitem{burke2013securing}
J.~Burke \emph{et~al.}, ``Securing instrumented environments over
  content-centric networking: the case of lighting control and {NDN},'' in
  \emph{INFOCOM WKSHPS}, 2013.

\bibitem{gusev2015ndn}
P.~Gusev \emph{et~al.}, ``{NDN-RTC}: Real-time videoconferencing over named
  data networking,'' in \emph{ICN}, 2015.

\bibitem{zhu2013let}
Z.~Zhu \emph{et~al.}, ``Let's chronosync: Decentralized dataset state
  synchronization in named data networking,'' in \emph{ICNP}, 2013.

\bibitem{carofiglio2011modeling}
G.~Carofiglio \emph{et~al.}, ``Modeling data transfer in content-centric
  networking,'' in \emph{ITC}, 2011.

\bibitem{garcia2014understanding}
J.~Garcia-Luna-Aceves \emph{et~al.}, ``Understanding optimal caching and
  opportunistic caching at the edge of information-centric networks,'' in
  \emph{ICN}, 2014.

\bibitem{yi2013case}
C.~Yi \emph{et~al.}, ``A case for stateful forwarding plane,'' \emph{Computer
  Communications}, vol.~36, no.~7, 2013.

\bibitem{rozhnova2014extended}
N.~Rozhnova \emph{et~al.}, ``An extended hop-by-hop interest shaping mechanism
  for content-centric networking,'' in \emph{GLOBECOM}, 2014.

\bibitem{carofiglio2012joint}
G.~Carofiglio \emph{et~al.}, ``Joint hop-by-hop and receiver-driven interest
  control protocol for content-centric networks,'' in \emph{SIGCOMM ICN
  Workshop}, 2012.

\bibitem{ren2015interest}
Y.~Ren \emph{et~al.}, ``An interest control protocol for named data networking
  based on explicit feedback,'' in \emph{ANCS}, 2015.

\bibitem{syverson1994taxonomy}
P.~Syverson, ``A taxonomy of replay attacks [cryptographic protocols],'' in
  \emph{CSFW}, 1994.

\bibitem{ferguson2000network}
P.~Ferguson, ``{RFC} 2827: Network ingress filtering: Defeating denial of
  service attacks which employ {IP} source address spoofing,'' 2000.

\bibitem{virgilio2013pit}
M.~Virgilio \emph{et~al.}, ``{PIT} overload analysis in content centric
  networks,'' in \emph{SIGCOMM ICN Workshop}, 2013.

\bibitem{compagno2013poseidon}
A.~Compagno \emph{et~al.}, ``Poseidon: Mitigating interest flooding {DDoS}
  attacks in named data networking,'' in \emph{LCN}, 2013.

\bibitem{you2014detecting}
R.~You \emph{et~al.}, ``Detecting and mitigating interest flooding attack in
  content centric networking,'' \emph{Advances in Computer Science and
  Technology}, vol.~65, 2014.

\bibitem{afanasyev2013interest}
A.~Afanasyev \emph{et~al.}, ``Interest flooding attack and countermeasures in
  named data networking,'' in \emph{IFIP Networking}, 2013.

\bibitem{nguyen2015detection}
N.~T. Nguyen \emph{et~al.}, ``Detection of interest flooding attacks in named
  data networking using hypothesis testing,'' 2015.

\bibitem{nguyen2015optimal}
T.~Nguyen \emph{et~al.}, ``An optimal statistical test for robust detection
  against interest flooding attacks in {CCN},'' in \emph{IFIP/IEEE IM}, 2015.

\bibitem{tang2013identifying}
J.~Tang \emph{et~al.}, ``Identifying interest flooding in named data
  networking,'' in \emph{GreenCom}, 2013.

\bibitem{mosko2013ccnx}
M.~Mosko, ``{CCNx} 1.0 protocol specification roadmap,'' 2013.

\bibitem{messages}
M.~Mosko \emph{et~al.}, ``{CCNx} messages in tlv format,'' 2015,
  \url{https://tools.ietf.org/html/draft-irtf-icnrg-ccnxmessages-00}.

\bibitem{stewart1998bgp4}
J.~W. Stewart~III, \emph{{BGP4}: inter-domain routing in the {Internet}}.\hskip
  1em plus 0.5em minus 0.4em\relax Addison-Wesley Longman, 1998.

\bibitem{waldvogel2007fast}
M.~Waldvogel, ``Fast longest prefix matching: algorithms, analysis, and
  applications,'' Ph.D. dissertation, ETH Z{\"u}rich, 2007.

\bibitem{kobayashi2000longest}
M.~Kobayashi \emph{et~al.}, ``A longest prefix match search engine for
  multi-gigabit ip processing,'' in \emph{ICC}, 2000.

\bibitem{dharmapurikar2003longest}
S.~Dharmapurikar \emph{et~al.}, ``Longest prefix matching using bloom
  filters,'' in \emph{SIGCOMM}, 2003.

\bibitem{ghali2014network}
C.~Ghali \emph{et~al.}, ``Network-layer trust in named-data networking,''
  \emph{SIGCOMM CCR}, vol.~44, no.~5, 2014.

\bibitem{ghali2014needle}
C.~Ghali \emph{et~al.}, ``Needle in a haystack: Mitigating content poisoning in
  named-data networking,'' in \emph{NDSS SENT Workshop}, 2014.

\bibitem{acs2013cache}
G.~Acs \emph{et~al.}, ``Cache privacy in named-data networking,'' in
  \emph{ICDCS}, 2013.

\bibitem{mohaisen2013protecting}
A.~Mohaisen \emph{et~al.}, ``Protecting access privacy of cached contents in
  information centric networks,'' in \emph{ASIA CCS}, 2013.

\bibitem{ndnSIM}
S.~Mastorakis \emph{et~al.}, ``{ndnSIM 2.0}: A new version of the {NDN}
  simulator for {NS-3},'' Technical Report, 2015.

\bibitem{ns3}
``Network simulator 3 {(NS-3)},'' \url{http://www.nsnam.org/}.

\bibitem{afanasyev2014nfd}
A.~Afanasyev \emph{et~al.}, ``{NFD} developers guide,'' Technical Report
  {NDN}-0021, {NDN} Project, Tech. Rep., 2014.

\bibitem{DFNverein}
``{DFN-Verein},'' \url{http://www.dfn.de/}.

\bibitem{DFN-NOC}
``{DFN-Verein: DFN-NOC},''
  \url{http://www.dfn.de/dienstleistungen/dfninternet/noc/}.

\bibitem{dai2013mitigate}
H.~Dai \emph{et~al.}, ``Mitigate ddos attacks in {NDN} by interest traceback,''
  in \emph{INFOCOM WKSHPS}, 2013.

\bibitem{al2015revisiting}
S.~Al-Sheikh \emph{et~al.}, ``Revisiting countermeasures against {NDN} interest
  flooding,'' in \emph{ICN}, 2015.

\bibitem{li2014interest}
Z.~Li \emph{et~al.}, ``Interest cash: an application-based countermeasure
  against interest flooding for dynamic content in named data networking,'' in
  \emph{CFI}, 2014.

\bibitem{bloom1970space}
B.~H. Bloom, ``Space/time trade-offs in hash coding with allowable errors,''
  \emph{ACM Communications}, vol.~13, no.~7, pp. 422--426, 1970.

\bibitem{wang2013decoupling}
K.~Wang \emph{et~al.}, ``Decoupling malicious interests from pending interest
  table to mitigate interest flooding attacks,'' in \emph{GC Wkshps}, 2013.

\bibitem{salah2015lightweight}
H.~Salah \emph{et~al.}, ``Lightweight coordinated defence against interest
  flooding attacks in {NDN},'' in \emph{INFOCOM WKSHPS}, 2015.

\bibitem{almishari2013optimizing}
M.~Almishari \emph{et~al.}, ``Optimizing bi-directional low-latency
  communication in named data networking,'' \emph{SIGCOMM CCR}, vol.~44, no.~1,
  2013.

\bibitem{dai2012pending}
H.~Dai \emph{et~al.}, ``On pending interest table in named data networking,''
  in \emph{ANCS}, 2012.

\bibitem{dibenedetto2012andana}
S.~DiBenedetto \emph{et~al.}, ``{ANDaNA}: Anonymous named data networking
  application,'' in \emph{NDSS}, 2012.

\end{thebibliography}
